\documentclass[a4paper,USenglish,cleveref, =autoref,thm-restate]{lipics-v2021}

\nolinenumbers

\pdfoutput=1 
\hideLIPIcs  

\title{Byzantine Consensus in Abstract MAC Layer} 

\author{Lewis Tseng}{Clark University, USA}{lewistseng@acm.org}{https://orcid.org/0000-0002-4717-4038}{}

\author{Callie Sardina}{UCSB, USA}{calliesardina@umail.ucsb.edu}{https://orcid.org/0000-0002-7498-2778}{}

\authorrunning{L. Tseng and C. Sardina} 

\Copyright{Lewis Tseng and Callie Sardina} 



\category{} 

\relatedversion{} 

\funding{This material is based upon work partially supported by the National Science Foundation under Grant CNS-2238020.}

\EventEditors{John Q. Open and Joan R. Access}
\EventNoEds{2}
\EventLongTitle{42nd Conference on Very Important Topics (CVIT 2016)}
\EventShortTitle{CVIT 2016}
\EventAcronym{CVIT}
\EventYear{2016}
\EventDate{December 24--27, 2016}
\EventLocation{Little Whinging, United Kingdom}
\EventLogo{}
\SeriesVolume{42}
\ArticleNo{23}

\usepackage{soul}
\usepackage{parskip}
\usepackage{algorithm}
\usepackage{algorithmic}
\usepackage{amsfonts,amsthm, amsmath}
\usepackage{multicol,color}
\usepackage[utf8]{inputenc}

\begin{document}

\maketitle

\begin{abstract}
This paper studies the design of Byzantine consensus algorithms in an \textit{asynchronous }single-hop network equipped with the ``abstract MAC layer'' [DISC09], which captures core properties of modern wireless MAC protocols. Newport [PODC14], Newport and Robinson [DISC18], and Tseng and Zhang [PODC22] study crash-tolerant consensus in the model. In our setting, a Byzantine faulty node may behave arbitrarily, but it cannot break the guarantees provided by the underlying abstract MAC layer. To our knowledge, we are the first to study Byzantine faults in this model. 

We harness the power of the abstract MAC layer to develop a Byzantine approximate consensus algorithm and a Byzantine randomized binary consensus algorithm. Both of our algorithms require \textit{only} the knowledge of the upper bound on the number of faulty nodes $f$, and do \textit{not} require the knowledge of the number of nodes $n$. This demonstrates the ``power'' of the abstract MAC layer, as consensus algorithms in traditional message-passing models require the knowledge of \textit{both} $n$ and $f$. Additionally, we show that it is necessary to know $f$ in order to reach consensus. Hence, from this perspective, our algorithms require the minimal knowledge. 
 
The lack of knowledge of $n$ brings the challenge of identifying a quorum explicitly, which is a common technique in traditional message-passing algorithms. A key technical novelty of our algorithms is to identify ``implicit quorums'' which have the necessary information for reaching consensus. The quorums are implicit because nodes do not know the identity of the quorums -- such notion is only used in the analysis. 
\end{abstract}

\newpage

\section{Introduction}
\label{s:intro}

We study the Byzantine consensus problems \cite{LamportSP82,DolevLPSW86,RandomizedBG_Rabin_SFCS_1983} in the ``abstract MAC layer'' model \cite{AbstractMAC_DISC2009,AbstractMAC_LEMIS_FOMC2012, AbstractMAC_probabilistic_AdHoc2014, AbstractMAC_unreliableLink_PODC2014}. The model was proposed by Kuhn, Lynch, and Newport  \cite{AbstractMAC_DISC2009} which harnesses the basic properties provided by existing wireless MAC (medium access control) protocols. The main purpose is to separate the high-level and low-level logic of algorithm design and the management of the wireless medium and participating nodes, respectively. 
Understanding the dynamics of these two levels helps one to explore the fundamental tradeoffs in algorithm design, and hopefully enables the  development and deployment of high-level algorithms onto low-level MAC protocols \cite{AbstractMAC_DISC2009,AbstractMAC_consensus_PODC2014,AbstractMAC_randomizedConsensus_DISC2018}.

The model is focused on an \textit{asynchronous} single-hop network in which nodes communicate via ``\textit{mac-broadcasts},'' the broadcast primitive provided by the abstract MAC layer. The mac-broadcast sends a message to all the fault-free nodes in the network, and the broadcaster will eventually receive an acknowledgement (ACK) upon the successful completion of the mac-broadcast. That is, upon learning the ACK, the broadcaster can be sure that all the fault-free nodes have received the message that was broadcast. The abstract MAC layer additionally provides authentication of messages between unknown processes. This primitive is stronger than the traditional point-to-point message-passing  \cite{attiya2004distributed, Lynch96} in the sense that prior works have proposed \textit{wait-free} crash-tolerant randomized consensus protocols \cite{AbstractMAC_randomizedConsensus_DISC2018,AbstractMAC_consensus_PODC2014} in the abstract MAC layer, which is typically infeasible in the point-to-point message-passing models \cite{attiya2004distributed, Lynch96}. 

To see the ``power'' of the abstract MAC layer, consider the case when node $i$ sends a message $m$ using mac-broadcast. Node $i$ does \textit{not} need to wait for explicit acknowledgement messages from other nodes. Instead, the abstract MAC layer provides the ACK which identifies the completion of the mac-broadcast. This indicates that all the other fault-free nodes have received the message $m$. In the case of traditional asynchronous point-to-point message-passing models, such guarantee is only ensured when node $i$ receives acknowledgement messages from all the other fault-free nodes, which is impossible when nodes may crash.  


Another important modeling choice is that prior works \cite{AbstractMAC_DISC2009,AbstractMAC_LEMIS_FOMC2012, AbstractMAC_probabilistic_AdHoc2014, AbstractMAC_unreliableLink_PODC2014,AbstractMAC_randomizedConsensus_DISC2018,AbstractMAC_consensus_PODC2014,Tseng_PODC22_MAC} assume little information known to the nodes to better capture the limitations of existing MAC protocols and the nature of the wireless networks. In particular, nodes do not have any a priori information about other nodes within the system. Some prior works (e.g., \cite{AbstractMAC_randomizedConsensus_DISC2018,AbstractMAC_consensus_PODC2014}) even study anonymous algorithms in which nodes do not have unique identifiers. We assume that each node has a unique identifier, but does not have prior information on the size of the system, or the identifiers of any other node at the beginning of the algorithm. 

\paragraph*{Consensus in Abstract MAC Layer} To our knowledge, there are three prior papers \cite{AbstractMAC_consensus_PODC2014,AbstractMAC_randomizedConsensus_DISC2018,Tseng_PODC22_MAC} that study fault-tolerant consensus in the abstract MAC layer. All three works consider either no failure, or assume only crash faults. In \cite{AbstractMAC_consensus_PODC2014},  Newport proves several impossibilities and identify consensus algorithms when there is no failure. Newport and Robinson 
\cite{AbstractMAC_randomizedConsensus_DISC2018} propose two algorithms that employ the abstract MAC layer to solve randomized binary consensus when nodes may crash. Tseng and Zhang \cite{Tseng_PODC22_MAC} provide another randomized consensus algorithm with improved storage complexity and expected time complexity. Approximate consensus is also studied in \cite{Tseng_PODC22_MAC}. 

We are not aware of any work on tackling Byzantine faults in the abstract MAC layer. We assume that a Byzantine faulty node can behave arbitrarily. It may also send inconsistent messages to other nodes. The only constraint of the Byzantine adversary is that it \underline{\textit{cannot}} break the underlying abstract MAC layer. First, it cannot ``nullify'' the delivery guarantee provided by the ACK at fault-free nodes. Second, even though nodes do not have information about other nodes a priori, communication is authenticated (by the MAC layer), and the receiver can verify the identity of the message sender, once it receives the message.  Consequently, a Byzantine adversary cannot fake its identity. From a more practical perspective, we are focused on the Byzantine faults at the application layer. For example, the Byzantine adversary under our consideration cannot use jamming or sybil attacks to undermine the abstract MAC layer.

The Fischer-Lynch-Paterson (FLP) impossibility result \cite{Impossibility_Fischer_ACM_1985} proves that it is impossible to design a deterministic exact consensus algorithm when nodes may fail in asynchronous  message-passing systems. The result can be extended to the abstract MAC layer model \cite{AbstractMAC_consensus_PODC2014}. Therefore, we focus on approximate consensus \cite{DolevLPSW86} and randomized binary consensus \cite{RandomizedBG_Rabin_SFCS_1983} problems. In the first problem, the agreement property is relaxed so that the outputs at fault-free nodes only need to be roughly equal, whereas in the second problem, the termination only holds in a probabilistic sense. 

\paragraph*{Our Contributions}

Consider an asynchronous system consisting of $n$ nodes and up to $f$ Byzantine faulty nodes. We propose the two following Byzantine-tolerant consensus algorithms: 

\begin{itemize}
    \item \textit{Approximate Consensus}: We present MAC-BAC, which is correct given $n \geq 5f+2$. Similar to prior approximate algorithms \cite{DolevLPSW86,AbrahamAD04}, nodes proceed in rounds and maintain a state value which is updated every round and eventually will become the output. MAC-BAC achieves convergence rate $3/4$. More concretely, after every two rounds, the range of the state values at fault-free nodes is reduced by at least $1/4$. 

    \item \textit{Randomized Consensus}: We present MAC-RBC, which is correct given $n \geq 5f+1$. The expected time complexity is constant. The algorithm assumes the existence of a common coin \cite{RandomizedBG_Rabin_SFCS_1983} among all fault-free nodes.
\end{itemize} 
Our algorithms require \textit{only} the knowledge of $f,$ and do \textit{not} need to know $n$. We also prove that the knowledge of $f$ is necessary for solving Byzantine consensus in Appendix \ref{app:impossible}. 

Our model is weaker than the synchronous point-to-point message-passing model \cite{LamportSP82}; hence, the lower bound on resilience $3f+1$ still applies. Moreover, a node can use mac-broadcast to simulate a point-to-point communication if $n$ is known. Therefore, prior algorithms \cite{AbrahamAD04,RandomizedBG_Rabin_SFCS_1983,Raynal_PODC14_optimalAsyncConsensus} with optimal resilience $3f+1$ can be simulated in our model with the knowledge of $n$. The lower bound on resilience when $n$ is unknown is left as an interesting future work.

The lack of knowledge of $n$ makes it impossible to identify a quorum explicitly, which is a common technique in consensus algorithms in the traditional message-passing models, e.g.,  \cite{attiya2004distributed,Lynch96,LamportSP82,DolevLPSW86,RandomizedBG_Rabin_SFCS_1983,Raynal_PODC14_optimalAsyncConsensus}. A key novelty of our algorithms is to identify ``implicit quorum.'' More precisely, our technique ensures that there exists a quorum whose information will eventually be propagated to other fault-free nodes; however, nodes do not know the explicit identifiers of the nodes inside the quorum. 

One challenge of such an implicit quorum is that the analysis becomes more complicated, as we first need to ensure that an implicit quorum exists and then we need to argue that other nodes will be able to learn the necessary information from the implicit quorum (which is not always obvious due to asynchrony and Byzantine faults).

\section{Related Work}

We discuss the most relevant consensus algorithms and the works on the abstract MAC layer. Modeling wireless networks with the abstract MAC layer was first introduced by Kuhn, Lynch, Newport in \cite{AbstractMAC_DISC2009}, in which  they present algorithms for multi-message broadcasts in a multi-hop network when there is no failure. Non-fault-tolerant leader election and maximal independent set problems are later studied in the model \cite{AbstractMAC_LEMIS_FOMC2012, AbstractMAC_probabilistic_AdHoc2014, AbstractMAC_unreliableLink_PODC2014}. 

The three prior works  \cite{AbstractMAC_consensus_PODC2014,AbstractMAC_randomizedConsensus_DISC2018,Tseng_PODC22_MAC} that study fault-tolerance in the abstract MAC layer all focus on crash faults. The techniques are different from our work, because Byzantine adversary can send inconsistent messages. For example, a technique of ``counter racing'' (for identifying when to output a value safely) is used in \cite{AbstractMAC_randomizedConsensus_DISC2018} and a technique of ``jumping'' to a state proposed by another node is proposed in \cite{Tseng_PODC22_MAC}. These techniques do not work if a Byzantine node lies about its observations or state value. 

The problem of Byzantine consensus in message-passing has been extensively studied in the literature since the seminal work by Lamport, Shostak and Pease \cite{LamportSP82}. Dolev et al. \cite{DolevLPSW86} propose an iterative approximate Byzantine consensus algorithm that is correct given $n \geq 5f+1$. Our algorithm MAC-BAC is inspired by their algorithm and requires $n \geq 5f+2$. Mostefaoui et al.  \cite{Raynal_PODC14_optimalAsyncConsensus} propose a Byzantine randomized binary consensus with optimal resilience $n \geq 3f+1$ and achieve expected constant time complexity. Our MAC-RBC algorithm is inspired by their algorithm, but requires $n \geq 5f+1$. Both algorithms use ``common coin'' \cite{RandomizedBG_Rabin_SFCS_1983}, which guarantees that every node receives the same sequence of random bits. Unlike prior algorithms, MAC-BAC and MAC-RBC does not use the information regarding $n$ or the notion of ``explicit quorum,'' so our design and analysis are more complicated than the ones in   \cite{DolevLPSW86,Raynal_PODC14_optimalAsyncConsensus}. 

Abrahm et al. \cite{AbrahamAD04} present an approximate consensus algorithm with optimal resilience. Their algorithm relies on the reliable broadcast primitive and the witness technique. Many works \cite{BAC_Bracha_1987, AuthenticatedFT_Srikanth_1987, Toueg_PODC_1984} describe Byzantine randomized binary consensus algorithms with various guarantees. These algorithms all achieve optimal resilience $n \geq 3f + 1$. There are also recent works on Byzantine randomized  consensus that require more powerful primitives such as PKI \cite{simpleRBC_Crain_2020,noPrivateSetupBC_Gaoz_IEEE_2022,subquadraticCommunicationBC_Blum_2020,concurrentAysnchBC_Cohen_2023,asymmetricBC_Cachin_ESORICS_2021}. All these algorithms require the knowledge of $n$ and rely on the usage of explicit quorum and some variation of reliable broadcasts \cite{bracha1987asynchronous}. 

Without knowing $n$, it is difficult to identify quorums explicitly so that there is an intersection between any two quorums. For example, in many prior works that use reliable broadcast (e.g.,  \cite{AbrahamAD04,bracha1987asynchronous}), a quorum of size $n-f$ is used, which ensures that a Byzantine node cannot equivocate. However, when $n$ is unknown, it is unclear whether such property can be guaranteed, forcing us to develop new techniques. In fact, the lower bound on the resilience of Byzantine consensus problems in the abstract MAC layer is still an open problem. 

There is also a line of works aiming to reach consensus in synchronous systems with unknown participants. The problem is named CUPs (Consensus with Unknown Participants). Similar to our model, the CUPs problem assumes no knowledge of $n$. It was first studied by Cavin et al. \cite{CUP_AdHoc-Now2004} when nodes do not crash. Greve and Tixeuil \cite{Greve_DSN_2007} study the tradeoffs between synchrony and the shared knowledge between nodes in a multi-hop network. Later, Alcheriri et al. \cite{ByzantineCUP_Alcheriri_OPODIS_2008} and Khanchandani and Wattenhofer \cite{Wattenhofer_ByzCUPs_PODC20} consider the Byzantine consensus in CUPs. These work assume synchrony; hence, are very different from our model.  

\section{Preliminaries}

\subsection{System Model}

Our system model consists of a static system with $n$ nodes, with up to $f$ nodes which may be Byzantine faulty. The set of nodes is denoted as the set of their unique identifiers, i.e., $\{1, \dots, n\}$. However, the knowledge of $n$ is only used in analysis. In our algorithms, nodes do not know $n$. Moreover, due to asynchrony and faults, it is impossible to learn $n$ exactly. 

Byzantine nodes may send arbitrary messages to other nodes, or act as crashed nodes. The messages which Byzantine nodes send to all other nodes need not be consistent. We assume that the behavior of the Byzantine nodes is controlled by a malicious adversary with access to the system state throughout the algorithm. Nodes which are not Byzantine are called fault-free nodes. Fault-free nodes follow the algorithm protocol. Our algorithm MAC-BAC assumes $n \geq 5f + 2$, and MAC-RBC assumes $n \geq 5f + 1$. 

Our algorithm operates on top of a single-hop network equipped with the abstract MAC layer \cite{AbstractMAC_kuhn_lynch_newport_2011}. The model provides a communication primitive ``mac-broadcast,'' which ensures an  eventual delivery guarantee. More specifically, at some point after a node $i$ has broadcast a message via ``mac-broadcast,'' node $i$ will receive an acknowledgment (ACK) which indicates that all other fault-free nodes within the system have received $i$’s message. No other information is contained within the ACK, e.g., the ACK relays no information concerning the number of other nodes within the system. As discussed in Section \ref{s:intro}, we consider Byzantine faults in the application layer; hence, the guarantees of the underlying abstract MAC layer cannot be disrupted by the Byzantine adversary.

\subsection{Approximate and Randomized Consensus}

A correct approximate consensus algorithm \cite{DolevLPSW86} must satisfy the following conditions: 

\begin{itemize}

    \item \textit{Termination}: Every fault-free node must output a value in a finite amount of time.

    \item \textit{Validity}: the output must remain in the convex hull of the inputs of the fault-free nodes. 

    \item \textit{$\epsilon$-Agreement}: For any $\epsilon > 0$, the output of all fault-free nodes are within $\epsilon$ of each other.
\end{itemize}

A correct randomized binary consensus  algorithm \cite{RandomizedBG_Rabin_SFCS_1983,Ben_Or_PODC_1983} must satisfy the following conditions when the input is a binary value (either $0$ or $1$): 

\begin{itemize}

    \item \textit{BC-Termination}: Every fault-free node outputs a value with probability 1.

    \item \textit{BC-Validity}: Every output value was proposed by a fault-free node.

    \item \textit{BC-Agreement}: The output of all fault-free nodes are identical.

\end{itemize}
\section{Byzantine Approximate Consensus: MAC-BAC}

This section presents our algorithm MAC-BAC, which is a correct Byzantine approximate consensus given $n \geq 5f+2$. It follows the structure of the algorithm by Dolev et al.  \cite{DolevLPSW86}. In both algorithms, node $i$ proceeds in rounds and keeps a state value $v_i$ that eventually becomes the output, after a sufficient number of rounds. The key difference between the two algorithms is that in MAC-BAC, node $i$ waits until it receives at least $4f+2$ messages from the same round (instead of $n-f$ in \cite{DolevLPSW86}). By assumption, node $i$ is able to transmit a message to itself using mac-broadcast. 

Recall that in our model, we assume nodes do not have the knowledge of $n$. Consequently, we do not have the notion of explicit quorum. (In \cite{DolevLPSW86}, the $n-f$ nodes from which a node $i$ received a message act as a quorum.) Therefore, our analysis is more complicated in the sense that we need to identify how important information is propagated throughout the rounds, via the help of ``implicit quorum.''

\subsection{MAC-BAC}

MAC-BAC is presented in Algorithm \ref{alg:MAC-BAC2}. The first step of the algorithm is to broadcast its identifier, its current value and round index using mac-broadcast. Once this mac-broadcast has completed, an ACK will be received from the abstract MAC layer acknowledging that the message has been received by all the fault-free nodes. 

Each node $i$ then waits to receive at least $4f+2$ messages from round $p_i$. Upon receiving these messages, node $i$ discards extreme values and update its new state value. We introduce two notations to facilitate the presentation:

\begin{itemize}
    \item $\min^{f+1}\{R_i[p_i]\}$ denotes the $(f+1)$-st minimum value in $R_i[p_i]$;\footnote{Alternatively, the smallest value after discarding $f$ smallest values in $R_i[p_i]$.} and 

    \item $\max^{f+1}\{R_i[p_i]\}$ denotes $(f+1)$-st maximum value in $R_i[p_i]$.\footnote{Alternatively, the largest value after discarding $f$ largest values in $R_i[p_i]$.}
\end{itemize}

Our strategy of updating the state value is as follows: at line 5, $l$ takes the $(f+1)$-st minimum value in $R_i[p_i]$. At line 6, $u$ takes the $(f+1)$-st maximum value in $R_i[p_i]$. The new state value at node $i$ is then updated to be the average of $l$ and $u$, at line 7. This is also the strategy used in \cite{AbrahamAD04,DolevDS87}.

Node $i$ then proceeds to the next round. Once node $i$ reaches the final round, $p_{end}$, it outputs the final state value, $v_i[p_{end} + 1]$.

\begin{algorithm}[t]
\caption{MAC-BAC: Steps at each node $i$}
\label{alg:MAC-BAC2}
\begin{algorithmic}[1]
\footnotesize
\item[{\bf Local Variables:}]
    
    \item[] $p_i$ \COMMENT{round index, initialized to $0$}
    
    \item[] $v_i$ \COMMENT{state, initialized to $x_i$, the input at node $i$}

    \hrulefill
    \begin{multicols}{2}    

    \FOR{$p_i \leftarrow 0$ to $p_{end}$}
        \STATE \textbf{mac-broadcast}$(i, v_i, p_i)$
        \STATE \textbf{wait until} node $i$ has received \\\hspace{18pt} $\geq 4f+2$ messages from round $p_i$
        \STATE $R_i[p_i] \leftarrow$ received round-$p_i$ messages
        \STATE $l = \min^{f+1}\{R_i[p_i]\}$
        \STATE $u = \max^{f+1}\{R_i[p_i]\}$
        \STATE $v_i[p_i+1] \leftarrow \frac{l + u}{2}$
        \STATE $p_i \leftarrow p_i + 1$
    \ENDFOR
    \STATE \textbf{output} $v_i[p_{end}+1]$

\end{multicols}  
\end{algorithmic}
\end{algorithm}

\subsection{Correctness Proof}

Termination is obvious, as  $p_{end}$ is a fixed value defined in Eq. (\ref{eq:p-end}). Moreover, since there are at least $5f+2$ nodes, each node is able to receive enough messages at Line 3. We present the proof in Appendix \ref{app:mac-bac-termination}. Validity also follows from the strategy of  discarding extreme values. Essentially, both $l$ and $u$ are guaranteed to be in the convex hull of the state values ($v_i$'s) at fault-free nodes from the previous round. The proof is presented in Appendix \ref{app:mac-bac-validity}. 

A key novelty is the way we prove $\epsilon$-agreement. In prior works  \cite{AbrahamAD04,DolevLPSW86}, the range of state values  at fault-free nodes shrinks every round, whereas in our proof, the range shrinks every \textit{two} rounds.  Moreover, in prior algorithms, any pair of two fault-free nodes must use at least one identical value to update their new state values, due to the usage of an explicit quorum. However, such a condition might not hold for MAC-BAC. This is because $n$ is unknown, and nodes might receive messages from two groups of nodes such that the intersection of the two group is less than $f$ nodes. In this case, there is no guarantee that nodes will use common value(s) to update the state values. In fact, in our algorithm, some nodes might use completely different values for updating (i.e., after discarding the common values) in the same round.

\subsubsection{Proof of $\epsilon$-Agreement and Implicit Quorum in MAC-BAC}

\paragraph*{Useful Notions}

We first introduce some terminology to facilitate the proof. 

\begin{definition}[First and Second Mover]
\label{def:first-and-second-movers}
    For each round $r$, the set of first movers is defined as the \underline{first $2f+1$ fault-free} nodes that complete their respective mac-broadcasts (at Line 2) in round $r$.\footnote{We can break ties using IDs without affecting the correctness.} All the other fault-free nodes are called second movers.
\end{definition}

In our analysis, we are interested in how first and second movers propagate and update their values. Therefore, we introduce two sets $F_r$ and $S_r$ below. 

\begin{definition}
    Let $F_r$ be the set of state values of the first movers at \underline{the end of round $r$} -- the $v_i$ \textbf{after} a first mover $i$ updates its state value at Line 7 in round $r$. 
    Let $S_r$ be the set of state values of the second movers at \underline{the end of round $r$} -- the $v_j$ \textbf{after} a second mover $j$ updates its state value at Line 7 in round $r$. 
\end{definition}

\begin{observation}[Sequential Order]
Without loss of generality, we can assume nodes complete Line 2 following a sequential order for each round. 
\end{observation}
For brevity of the presentation, we relabel the IDs so that node $j$ completes before node $i$.\footnote{Assuming that nodes complete Line 2 in a sequential order within each round does not affect the correctness because nodes only process messages received from nodes at the same round, and the state of a node changes only once within a round. When the state updates at Line 7, the round number increments as well on Line 8 (only round p messages are processed in round p, so any change of states is the value for the subsequent round). Ordering does not alter the values sent to/ from nodes, nor does it allow for values to be considered from the incorrect round. Therefore, correctness is not violated, after the ID relabeling.} 

We know by sequential ordering that if $j < i$, then $j$ completes its mac-broadcast before $i$. This means that node $j$ must have received its ACK from the mac-broadcast before node $i$ completes Line 2. Therefore, in order to move to the next round $r+1$, $i$ must receive node $j$'s round-$r$ state, i.e., $v_j[r]$ that is assigned at Line 5. Wth a slight abuse of terminology, let node's round-0 state be the input for that node. 

\begin{observation}
\label{obs2}
Following the sequential order and the property of the mac-broadcast, we know that for each round $r+1$, node $i$ must receive node $j$'s round-$r$ state if $j < i$ for all fault-free $i$ and $j$. 
\end{observation}

Note that by definition, if $j$ is fault-free, then it is either a first or second mover. Additionally, this observation does not indicate the relationship between $R_i[r+1]$ and $R_j[r+1]$. In particular, it is possible that $v_k[r] \in R_i[r+1]$ and $v_k[r] \notin R_j[r+1]$. This is possible if $i, j < k$ or $j < k < i$.

Observation \ref{obs2} and the guarantees of the abstract MAC layer together imply that the state values broadcast by the first movers are received by all the second movers. 

\paragraph*{Implicit Quorum in MAC-BAC}

In our analysis, first movers are the ``implicit quorum'' for second movers in round $r$, due to Observation \ref{obs2}. This is because even though second movers do not know the identities of the first movers, second movers will share the same information from the first movers and use some of the state values at first movers to update their new state values.

Interestingly, first movers may not have enough shared information in round $r$ \textit{within themselves}. This is possible if they receive many messages from non-overlapping sets of second movers at Line 3. They are only guaranteed to receive common information from their ``implicit quorum'' in the next round. More concretely, first movers of round $r+1$ are guaranteed to receive enough information (for convergence) from the second movers of round $r$. This is because node $i$ waits for $4f+2$ messages. Among them, $2f+1$ could be from first movers of round $r$, $f$ could be from Byzantine nodes, and the remaining $f+1$ could be from second movers of round $r$. This turns out is enough for first movers of round $r+1$ to converge. The proof of Lemma \ref{lemma:interval-length} presents this intuition in more detail. 

\paragraph*{Proof of $\epsilon$-Agreement}

Without loss of generality, we can scale the inputs to $[0,1]$ as long as we scale $\epsilon$ down by the same factor. For simplicity of the presentation, we assume that for each fault-free node $i$, its input $x_i\in [0,1]$.

We first prove the following lemma. 
The lemma below follows from the fact that each node discards extreme values from Byzantine nodes. The proof is presented in Appendix \ref{app:range}.

\begin{lemma}
\label{lemma:f1}
    Fix a round $r \geq 1$. Assume the range of state values at fault-free nodes is $[x, y]$, where $0 \leq x, y \leq 1$, i.e., $F_r \cup S_r = [x, y]$. Then, we have $F_{r+1} \subseteq F_r \cup S_r = [x, y]$.
\end{lemma}

Using the same argument, we can also show that $F_1 \subseteq [0, 1],$ the range of input $x_i$.

Before proving how the state values at second movers evolve, we first introduce two notations for a round $r \geq 1$: 

    \begin{itemize}
        \item Let $m_{r}$ be the minimum fault-free state value at the end of round $r$, $m_r = \min\{F_r \cup S_r\}$.
        \item Let $M_{r}$ be the maximum fault-free state value at the end of round $r$, $M_r = \max\{F_{r} \cup S_{r}\}$.
    \end{itemize}

It follows that the interval length of $F_r \cup S_r$ is $M_r - m_r$. 

Then, we prove the following lemma. The lemma is where we utilize the ``implicit quorum'' for second movers. By the property of mac-broadcast, every second mover must receive the state values from all the first movers (Observation \ref{obs2}); hence, we can use this observation to show that the interval length of $S_{r+1}$ must shrink by at least half. This particular proof is similar to the ones in traditional message-passing networks \cite{DolevLPSW86,AbrahamAD04}. 

\begin{lemma}
\label{lemma:1/2}
    Fix round $r\geq 1$. The interval length of $S_{r+1}$ is at most half of the interval length of $F_r \cup S_r$.
\end{lemma}

\begin{proof}

    Let $x_{r+1}$ be the median of the state values at first movers in \underline{round $r+1$.} 
    By definition, we have $m_{r} \leq x_{r+1}$ and $M_{r} \geq x_{r+1}$.

    Now, consider any two second movers $i$ and $j$. Without loss of generalization, assume $v_i[r+2] \geq v_j[r+2]$. Recall that these values are produced at the end of round $r+1$ at Line 7.

    Since $i$ discards extreme values, $u \leq M_{r}$ and $l \leq x_{r+1}$ at node $i$. Similarly, $x_{r+1} \leq u$ and $m_{r} \leq l$ at node $j$. Therefore, we have in round $r+1$,

    \[
    v_i[r+2] = \frac{l+u}{2} \leq \frac{x_{r+1} + M_{r}}{2}
    \]

    and 

    \[
    v_j[r+2] = \frac{l+u}{2} \geq \frac{m_{r}+x_{r+1}}{2}
    \]

    Consequently, we have

    \[
    v_i[r+2] - v_j[r+2] \leq \frac{M_{r} + x_{r+1}}{2} - \frac{x_{r+1}+m_{r}}{2}
    = \frac{M_{r}-m_{r}}{2}
    \]

    Since the inequality applies to any pair of second movers $i$ and $j$, the interval length of $S_{r+1}$ is at most half of the interval length of $F_r \cup S_r$. (Note that the interval length of $F_r \cup S_r$ is simply $M_r - m_r$.)

\end{proof}

The proof can be easily applied to the case of $S_1$. That is, the interval length of $S_1$ is at most half of the interval length of the inputs, $[0, 1]$.

We then prove the following key lemma. This proof is where we use the notion of implicit quorum for first movers. In particular, first movers in round $r+1$ rely on second movers in round $r$ to learn the information that is essential for convergence. 

\begin{lemma}
\label{lemma:interval-length}
    The interval length of $F_{r+2} \cup S_{r+2}$ is at most $\frac{3}{4}(M_r-m_r)$. 
\end{lemma}

\begin{proof}

First, let us define

\begin{itemize}
    \item $F_{r+1} = [a^F_{r+1}, b^F_{r+1}]$

    \item $S_{r+1} = [a^S_{r+1}, b^S_{r+1}]$
\end{itemize}
Note that all these four values (at a respective bound) are in the range of $[m_r, M_r]$ due to Lemma \ref{lemma:f1}. 

Now, we consider the smallest possible value for $F_{r+2} \cup S_{r+2}$.

\begin{itemize}
    \item Case I: if $a^F_{r+1} < a^S_{r+1}$. 
    
    In this case, the smallest value is $\frac{a^F_{r+1}+a^S_{r+1}}{2}$. This is because there are at most $2f+1$ $a^F_{r+1}$ in $F_{r+1}$ and up to $f$ Byzantine nodes can send values $\leq a^F_{r+1}$. The remaining $f+1$ values must come from $S_{r+1}$ whose smallest value is $a^S_{r+1}$. After discarding $f$ values, at least a value that is $\geq a^S_{r+1}$ remains to be used to update the state value at Line 7.

    \item Case II: if $a^F_{r+1} \geq a^S_{r+1}$. 
    
    In this case, the smallest value is $a^S_{r+1}$. This is because there could be more than $4f+2$ $a^S_{r+1}$'s in $S_{r+1} \cup F_{r+1}$. 
\end{itemize}

Next, we consider the largest possible value for $F_{r+2} \cup S_{r+2}$. 

\begin{itemize}
    \item Case III: if $b^F_{r+1} > b^S_{r+1}$.
    
    In this case, the largest value is $\frac{b^F_{r+1}+b^S_{r+1}}{2}$. This is because there are at most $2f+1$ $b^F_{r+1}$ in $F_{r+1}$ and up to $f$ Byzantine nodes can send values $\geq b^F_{r+1}$. The remaining $f+1$ values must come from $S_{r+1}$ whose largest value is $b^S_{r+1}$. After discarding $f$ values, at least a value that is $\leq b^S_{r+1}$ remains to be used to update the state value at Line 7.

    \item Case IV: if $b^F_{r+1} \leq b^S_{r+1}$.
    
    In this case, the largest value is $b^S_{r+1}$. This is because there could be more than $4f+2$ $b^S_{r+1}$'s in $S_{r+1}$. 
\end{itemize}

Now, we can consider the following four cases to bound the interval length of $F_{r+2} \cup S_{r+2}$:

\begin{itemize}
    \item $a^F_{r+1} < a^S_{r+1}$ and $b^F_{r+1} > b^S_{r+1}$:

    The interval length is

    \begin{align*}
    \frac{b^F_{r+1}+b^S_{r+1}}{2} - \frac{a^F_{r+1}+a^S_{r+1}}{2} &=    \frac{1}{2}\{(b^F_{r+1}-a^F_{r+1}) + (b^S_{r+1}-a^S_{r+1})\}\\
    &= \frac{1}{2}\{(M_r-m_r)+(M_r-m_r)/2\} 
    = \frac{3(M_r-m_r)}{4}
    \end{align*}
    
    \item $a^F_{r+1} < a^S_{r+1}$ and $b^F_{r+1} \leq b^S_{r+1}$:

    The interval length is

    \begin{align*}
        b^S_{r+1} - \frac{a^F_{r+1}+a^S_{r+1}}{2}  &= \frac{1}{2}\{(b^S_{r+1}-a^S_{r+1})+(b^S_{r+1}-a^F_{r+1})\}\\
        &= \frac{1}{2}\{(M_r-m_r)/2+(M_r-m_r)\} = \frac{3(M_r-m_r)}{4} 
    \end{align*}
    
    \item $a^F_{r+1} \geq a^S_{r+1}$ and $b^F_{r+1} > b^S_{r+1}$:

    The interval length is

    \begin{align*}
    \frac{b^F_{r+1}+b^S_{r+1}}{2} - a^S_{r+1} &=    \frac{1}{2}\{(b^F_{r+1}-a^S_{r+1}) + (b^S_{r+1}-a^S_{r+1})\}\\
    &= \frac{1}{2}\{(M_r-m_r)+(M_r-m_r)/2\} 
    = \frac{3(M_r-m_r)}{4}
    \end{align*}

    \item $a^F_{r+1} \geq a^S_{r+1}$ and $b^F_{r+1} \leq b^S_{r+1}$:

    The interval length is

    \[
    b^S_{r+1} - a^S_{r+1} = \frac{(M_r-m_r)}{2}
    \]
\end{itemize}

In each case above, we saw that compared to $F_r \cup S_r$, the interval length of $F_{r+2} \cup S_{r+2}$ shrinks by at least $1/4$, proving the lemma.
\end{proof}

Now, we are ready to prove that MAC-BAC converges with the desirable convergence rate. 

\begin{theorem}
    MAC-BAC achieves $\epsilon$-agreement in 
    \begin{equation}
    \label{eq:p-end}
        p_{end} \geq 2 \cdot \log_{\frac{3}{4}}\epsilon 
    \end{equation}
    rounds.
\end{theorem}

\begin{proof}
    By the conclusion of Lemma \ref{lemma:interval-length}, $F_{r} \cup S_{r} \leq \frac{3}{4} \cdot F_{r-2} \cup S_{r-2}$. Therefore, to satisfy $\epsilon$-agreement, the number of iteration $r$ must satisfy the following inequality.

    \begin{align*}      
        \epsilon &\geq \frac{3}{4}^{\lfloor \frac{r}{2} \rfloor}\\
        \log \epsilon &\geq \lfloor \frac{r}{2} \rfloor \cdot \log \frac{3}{4}\\
        2\cdot \log \epsilon &\geq r \cdot \log \frac{3}{4}\\
        \frac{2\cdot \log \epsilon}{\log \frac{3}{4}} &\leq r\\
        r &\geq 2\cdot \log_{\frac{3}{4}} \epsilon 
    \end{align*}

    Therefore, $\epsilon$-agreement will be achieved in $\geq 2\cdot \log_{\frac{3}{4}} \epsilon$ rounds. If we define $p_{end}$ as the smallest integer that satisfies the inequality in Algorithm MAC-BAC, then $\epsilon$-agreement is achieved. 

\end{proof}

\section{Byzantine Randomized Binary Consensus: MAC-BRC}

Assuming $n \geq 5f + 1$, our algorithm MAC-BRC correctly solves Byzantine randomized binary consensus. In MAC-BRC, each node is assumed to have a common coin provided by a trusted dealer, as in the work by Rabin \cite{RandomizedBG_Rabin_SFCS_1983}, which guarantees that every node shares the same sequence of random bits $b_1, b_2, \cdots , b_k, \cdots$ with value $0$ or $1$, each with probability $\frac{1}{2}$. Additionally, the common coin is ``global,'' meaning that the $k$-th call to $coinflip()$ (Line 9 of Algorithm MAC-RBC) by a fault-free node will return the same bit, $b_c$, to all nodes invoking the $k$-th $coinflip()$. 

Our algorithm is inspired by \cite{Raynal_PODC14_optimalAsyncConsensus}, especially the way we use the common coin to decide whether it is safe to output a value. As mentioned earlier, the key technical contribution is the usage of implicit quorum, which will become clear when we present the algorithm. 

\subsection{MAC-RBC}

MAC-RBC is presented in Algorithm \ref{alg:simple-RBC}. 
Nodes proceed in phases. In phase $p$, each node $i$ first does a mac-broadcast of an $EST$ message containing $i$'s ``estimated'' value (that $i$ estimates to be the output based on the information it collected in the previous phase) and phase number to other nodes. Once this mac-broadcast has completed, an ACK will be received. Meanwhile, a background event handler processes all $EST$ messages to determine when a value can be safely added to a local estimated set, $estValues_i[p]$. The way the handler is constructed ensures that the value added $estValues_i[p]$ must be an estimate value $v_i$ by some fault-free node. 

Once $estValues_i[p]$ is non-empty, the main thread resumes at Line 4, where node $i$ mac-broadcasts an $AUX$ message for this value at node $i$. Another background event handler processes $AUX$ messages, and adds the identifiers of all nodes which sent $AUX$ messages (for a particular value $w$) to its set $U_i$. $U_i$ is used to count the number of other nodes that supports a certain value $w$. Node $i$ then mac-broadcasts a $COMPLETE$ message indicating that it has completed its broadcasting of an $AUX$ message. 

Once our Condition WAIT (defined below) -- which ensures that a node only adds fault-free values which have been sent and received by sufficiently many other nodes -- is satisfied, node $i$ will have some value(s) in its set $values_i[p]$. A call to $coinflip()$ on Line 9 employs the global common coin whose returned value $c$ is then compared to $values_i[p]$. If these values are equal, that is $values_i[p] = \{v\} = c$, node $i$ will output $v$. Otherwise, node $i$ will adopt the value of the common coin and continue to the subsequent phase until it outputs a value.  

The key novelty is the construction of Condition WAIT, defined in Definition \ref{def:wait}. The condition makes sure that enough information is shared between any pair of fault-free nodes upon the satisfaction of Condition WAIT. On a high-level, the condition relies on two key elements: (i) a counter set $U_i[p]$ that keeps track of nodes that $i$ knows; and (ii) an ``implicit quorum'' $X$ which contains the nodes that saw the same estimate value $v$ and have completed their mac-broadcast. Note that $X$ is implicit in the sense that the $X$ at node $i$ might not always intersect with the $X$ at node $j$. However, it turns out that it is already sufficient for our purpose. (More details in the proof of BC-Agreement in Section \ref{s:BC-Agreement}.)

\begin{definition}[Condition WAIT]
\label{def:wait}
    In phase $p$, node $i$ satisfies Condition WAIT with value $v$, if there exist two sets of nodes $X$ and $Y$ such that

    \begin{enumerate}
    
        \item $|X| \geq 2f+1$;
        
        \item $i$ received $(COMPLETE, p)$ message for each $x\in X$;

        \item $i$ received $(AUX, v, x, p)$ message for each $x\in X$ and some identical value $v$; 
        
        \item $|Y| = |U_i[p]|-f$;
        
        \item $i$ received $(AUX, *, y, p)$ message for each $y\in Y$;\footnote{Note that these $AUX$ messages might not contain $v$.}
        
        \item $X \subseteq Y$;
        
        \item let $value_i[p]$ be the set of values contained in $Y$'s $(AUX, *, *, p)$ messages; 
        
        \item $value_i[p] \subseteq estValues_i[p]$.\footnote{Note that $estValues_i[p]$ could keep growing even after the execution of Line 3 to Line 5, as the background message handler is long-living.}
        
    \end{enumerate}
\end{definition}

Figure \ref{fig:wait} illustrates the relation between set $Y$ and set $X$ specified in Condition WAIT. 

\begin{algorithm}[t]
\caption{Mac-RBC: Steps at each node $i$}
\label{alg:simple-RBC}
\begin{algorithmic}[1]
\footnotesize
\item[{\bf Local Variables:}]
    
    \item[] $p_i$ \COMMENT{phase, initialized to $0$}
    
    \item[] $v_i$ \COMMENT{state, initialized to $x_i$, the input at node $i$} 

    \item[] $estValues_i[p]$ \COMMENT{set, initialized to $\{\}$} 

    \item[] $U_i[p]$ \COMMENT{set, initialized to $\{\}$} 
    
    
    \item[] \hrulefill
    \begin{multicols}{2}    

    \WHILE{true}
        
        \STATE \textbf{mac-broadcast}$(EST, v_i, p_i)$

        \STATE \textbf{wait until} $estValues_i[p_i] \neq \emptyset$

        \FOR{each $w \in estValues_i[p_i]$}
            \STATE \textbf{mac-broadcast}$(AUX, w, i,  p_i)$
        \ENDFOR

        \STATE \textbf{mac-broadcast}$(COMPLETE, p_i)$
 
        \STATE \textbf{wait until} Condition WAIT is satisfied \\\hspace{20pt}with some value $z$
        

        \STATE $c \gets coinflip()$

        \IF{$values_i[p_i] = \{v\}$}
            \IF{$v = c$}
                \STATE output $v$
            \ENDIF
            \STATE $v_i \gets v$
        \ELSE
            \STATE $v_i \gets c$     
        \ENDIF

        \STATE $p_i \gets p_i+1$
        
    \ENDWHILE

    \columnbreak
    \item[//\textit{Background EST message handler}]
    \item[\textbf{Upon} receiving ($EST, v, p$) \textbf{do}] 
    \IF{($EST, v, p$) is received from $f + 1$ nodes and ($EST, v, p$) not yet broadcast by $i$}
        \STATE \textbf{mac-broadcast}($EST, v, p$))
    \ENDIF
    \IF{($EST, v, p$) is received from $2f + 1$ nodes}
        \STATE $estValues_i[p] \gets estValues_i[p] \cup \{v\}$
    \ENDIF
    
    \item[]
    
    \item[//\textit{Background AUX message handler}]
    \item[\textbf{Upon} receiving $(AUX, *, j, p)$ \textbf{do}]
    \STATE $U_i[p] \gets U_i[p] \cup \{j\}$\COMMENT{Even if $j$ sends two different $AUX$ msgs, $j$ is added only once}
\end{multicols}  
\end{algorithmic}
\end{algorithm}

\begin{figure}[t]
    \centering
    \includegraphics[width=\textwidth]{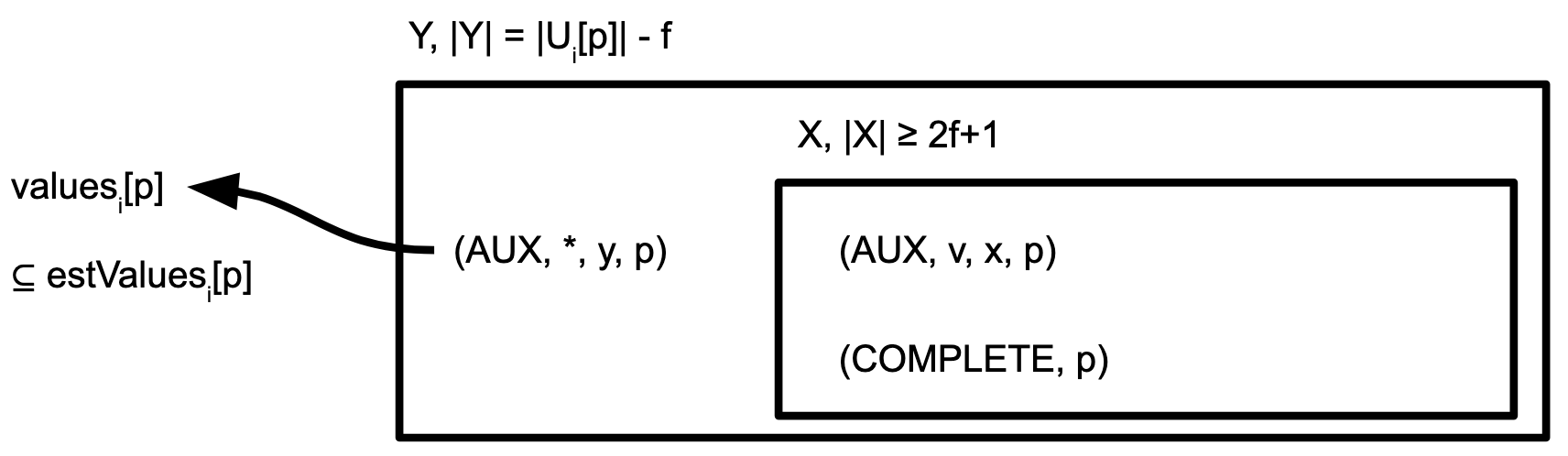}
    \caption{Illustration of Condition WAIT}
    \label{fig:wait}
\end{figure}

\subsection{Correctness Proof}

The BC-Validity proof follows from the construction of $estValues_i[p]$ (the $f+1$ threshold), and the observation that the output must be some value from $estValues_i[p]$. 

\begin{theorem}
Given that $n \geq 5f+1$, MAC-RBC satisfies BC-Validity. 
\end{theorem}

\begin{proof}
    Fix a phase $p$ and let node $i$ be a fault-free node with value $v \in values_i$ which has been mac-broadcast as an estimate value 
    by a fault-free node. By the wait statement at Line 3, since each fault-free node $i$ mac-broadcasts the values within its set $estValues_i$, and by the Condition WAIT, the set $values_i$ contains only values from fault-free nodes. The set $estValues_i$ contains only values from fault-free nodes because all values added to $estValues_i$ must have been sent by at least $f+1$ nodes on in order to pass Line 20. There are at most $f$ Byzantine nodes, so one of these $(EST, v, p)$ messages
    must have been broadcast by a fault-free node.
    
    If $values_i = \{v\} = c$, the value of the common coin, node $i$ outputs $v$ at Line 13 and sets its estimate value to $c$. If $values_i = \{v, v'\}$, both values have been processed by fault-free nodes, and node $i$ adopts the value of the common coin as the estimate value for node $i$ in phase $p_i + 1$ at Line 17. In both cases, the estimate value of a fault-free node is a value that has been proposed by a fault-free node. 
\end{proof}

The BC-Termination proof is focused on showing that as long as $n \geq 5f+1$, then Condition WAIT can always be satisfied under all possible scenarios. Moreover, the termination with probability 1 roughly follows the proof structure in \cite{Raynal_PODC14_optimalAsyncConsensus}, which relies on the usage of common coin and the cardinality of $value_i[p]$ (the condition to check at Line 10). The full proof is presented in Appendix \ref{app:bc-termination}. 

\begin{theorem}
Given that $n \geq 5f+1$, MAC-RBC terminates with probability 1. 
\end{theorem}

\subsubsection{Proof of BC-Agreement and Implicit Quorum in MAC-RBC}
\label{s:BC-Agreement}

Let us define $values_i^r[p]$ to be the set of values $values_i[p]$ right after node $i$ completes Line 8. That is, the $values_i[p]$ that node $i$ derived from Condition WAIT (Definition \ref{def:wait}). As mentioned earlier, $X$ identified in Condition WAIT could be different for two different nodes. This is because $n$ is unknown, and two nodes might use different sets of $2f+1$ nodes as $X$. However, in the proof of lemma below, we demonstrate that under a certain case, $X$ at node $i$ is guaranteed to intersect with $Y$ at node $j$. This is mainly the usage of the COMPLETE message. Even though this claim does \textit{not} imply that $X$ at node $i$ will intersect with $X$ at node $j$; however, due to the usage of a common coin, this is already enough for showing agreement of MAC-RBC. 

\begin{lemma}
\label{lemma:v=u}
    Fix a phase $p$. For any fault-free $i$ and $j$ with $values_i^r[p] = \{v\}$ and $values_j^r[p] = \{u\}$, then $v = u$.
\end{lemma}

\begin{proof}
    Assume node $i$ completes line 8 at time $T_i$ with $values_i^r[p] = \{v\}$. By construction, node $i$ has Condition WAIT satisfied with value $v$.  

    From conditions (1), (4), (6) and (7) of Condition WAIT, node $i$ has received $AUX$ messages 
    from a set $Y_i$ of size $|U_i[p]|-f \geq 2f+1$. (At this point, $estValues_i[p]$ might contain some value other than $v$, but we do not care about it.) We first prove the following claim.

    \begin{claim}
        \label{claim:RBC}
        At least $f+1$ fault-free nodes in $Y_i$ have completed line 5 with value $v$ in the $AUX$ message before time $T_i$. That is, at least $f+1$ fault-free nodes have mac-broadcast($AUX, v, y, p$) for some $y \in Y_i$ by time $T_i$.
    \end{claim}

    \begin{proof}[Proof of Claim \ref{claim:RBC}]
         By condition (3) and (6) in Condition WAIT, $Y_i$ contains $X_i$, and every fault-free node in $X_i$ have completed line 5 with value $v$ before broadcasting the $COMPLETE$ message. Since $|X_i| \geq 2f+1$, and up to $f$ nodes can be Byzantine, at least $f+1$ nodes in $X_i$ has completed line 5 with value $v$ before time $T_i$, proving the claim. 
         
    \end{proof}
    
    Let us denote the set of fault-free nodes identified in Claim \ref{claim:RBC} by $Y^v$. (Note that $Y_v$ is a superset of $X_i$)

    Now consider node $j$ with $values_j^r[p] = \{u\}$. Without loss of generality, assume that $j$ completes line 8 at some later time $T_j$, i.e., $T_j \geq  T_i$. This assumption together with Claim \ref{claim:RBC} imply that $U_j[p]$ contains $Y^v$ at time $T_j$, due to the guarantee of mac-broadcast. 
    
    By the definition of Condition WAIT, $Y_j$ contain $|U_j[p]| - f$ nodes at time $T_j$. This implies that the intersection of $Y_j$ and $Y^v$ is non-empty. This is because the size of $Y^v$ is at least $f+1$.  Therefore, value $v$ must be in $value_j[p]$, i.e., $v \in value_j[p_j]$.  Since $value_j[p]$ contains a single element, this  implies that $v=u$.
\end{proof}


\begin{theorem}
Given $n \geq 5f+1$, MAC-RBC satisfies BC-Agreement. 
\end{theorem}

\begin{proof}
    Let phase $p$ be the first phase at which a fault-free node $i$ outputs a value $v$ at Line 12. For any $j$ that also outputs a value in phase $p$, both $i$ and $j$ must output the same value, namely, the value of the common coin. 

    Consider any node $j$ that has not output in phase $p$. Observe that we have $values_i[p] = \{v\}$. It is then impossible by Lemma \ref{lemma:v=u} for $values_j[p] = \{v'\}$, with $v' \neq v$. Therefore, $values_j = \{v, v'\}$. Node $j$ will then execute Line 16 and set $v_j$ to be the value of the common coin in phase $p$. By construction, this value is $v$.
    
    Then, node $j$'s estimate value in phase $p + 1$ will be $v$. The output values must be an estimate value, and in phase $p+1$, all fault-free nodes have the same estimate value and will hence output $v$. This proves that the BC-Agreement property of MAC-RBC.
\end{proof}

\section{Impossibility}

In this section, we provide the intuition for our proof that without the knowledge of $f$, it is impossible to solve consensus. The full proof is included in Appendix \ref{app:impossible}.

We construct an indistinguishably proof by constructing scenarios in which there is no way for a node $i$ to distinguish whether another node $j$’s behavior is Byzantine or not, which could lead to a violation of validity. We assume the existence of an algorithm $A$ which solves consensus for a certain $n$ and $f$ without the knowledge of these values. We construct two scenarios with different value for $n$ and $f$. These scenarios remain indistinguishable to all nodes because no nodes have knowledge of $n$ or $f$. In the first scenario, node $i$ can only communicate with one other node to update its state, and in the second, node $i$ can receive messages from $\geq 2f+1$ other nodes within the system. In the second scenario, we impose a time delay, $D$, on all messages sent from any node other than an arbitrary node $j$. We observe that within the time interval $(0, D]$, a node $i$ cannot distinguish between the two scenarios (whether there is only one, or more than one other nodes within the system). Now, node $i$ has no way of determining whether the behavior of node $j$ is Byzantine or fault-free. Node $i$ then runs algorithm $A$, and may output a value that is outside the range of fault-free inputs if $i$ considers a Byzantine node $j$’s value, hence, violating validity.

\section{Summary}

This paper studies Byzantine consensus problems in the abstract MAC layer. We present MAC-BAC, a Byzantine approximate consensus algorithm, and MAC-RBC, a Byzantine randomized binary consensus algorithm. Both algorithms do not require the knowledge of $n$. To achieve so, we rely on the notion of implicit quorum. Therefore, our analysis is sufficiently different from prior work. One interesting open problem is the lower bound on the resilience of Byzantine consensus algorithms. 




\bibliography{references, reference-ipdps22, Tseng}

\begin{thebibliography}{10}

\bibitem{AbrahamAD04}
Ittai Abraham, Yonatan Amit, and Danny Dolev.
\newblock Optimal resilience asynchronous approximate agreement.
\newblock In {\em Principles of Distributed Systems, 8th International
  Conference, {OPODIS} 2004, Grenoble, France, December 15-17, 2004, Revised
  Selected Papers}, pages 229--239, 2004.
\newblock \href {https://doi.org/10.1007/11516798\_17}
  {\path{doi:10.1007/11516798\_17}}.

\bibitem{ByzantineCUP_Alcheriri_OPODIS_2008}
Eduardo~A. Alchieri, Alysson~Neves Bessani, Joni Silva~Fraga, and Fab\'{\i}ola
  Greve.
\newblock Byzantine consensus with unknown participants.
\newblock In {\em Proceedings of the 12th International Conference on
  Principles of Distributed Systems}, OPODIS '08, page 22–40, Berlin,
  Heidelberg, 2008. Springer-Verlag.
\newblock \href {https://doi.org/10.1007/978-3-540-92221-6_4}
  {\path{doi:10.1007/978-3-540-92221-6_4}}.

\bibitem{attiya2004distributed}
Hagit Attiya and Jennifer Welch.
\newblock {\em Distributed computing: Fundamentals, Simulations, and Advanced
  topics}, volume~19.
\newblock John Wiley \& Sons, 2004.

\bibitem{Ben_Or_PODC_1983}
Michael Ben-Or.
\newblock Another advantage of free choice (extended abstract): Completely
  asynchronous agreement protocols.
\newblock In {\em Proceedings of the Second Annual ACM Symposium on Principles
  of Distributed Computing}, PODC '83, page 27–30, New York, NY, USA, 1983.
  Association for Computing Machinery.
\newblock \href {https://doi.org/10.1145/800221.806707}
  {\path{doi:10.1145/800221.806707}}.

\bibitem{subquadraticCommunicationBC_Blum_2020}
Erica Blum, Jonathan Katz, Chen-Da Liu-Zhang, and Julian Loss.
\newblock Asynchronous byzantine agreement with subquadratic communication.
\newblock Cryptology ePrint Archive, Paper 2020/851, 2020.
\newblock \url{https://eprint.iacr.org/2020/851}.
\newblock URL: \url{https://eprint.iacr.org/2020/851}.

\bibitem{BAC_Bracha_1987}
Gabriel Bracha.
\newblock Asynchronous byzantine agreement protocols.
\newblock {\em Information and Computation}, 75(2):130--143, 1987.
\newblock URL:
  \url{https://www.sciencedirect.com/science/article/pii/089054018790054X},
  \href {https://doi.org/https://doi.org/10.1016/0890-5401(87)90054-X}
  {\path{doi:https://doi.org/10.1016/0890-5401(87)90054-X}}.

\bibitem{bracha1987asynchronous}
Gabriel Bracha.
\newblock Asynchronous {B}yzantine agreement protocols.
\newblock {\em Information and Computation}, 75(2):130--143, 1987.

\bibitem{asymmetricBC_Cachin_ESORICS_2021}
Christian Cachin and Luca Zanolini.
\newblock Asymmetric asynchronous byzantine consensus.
\newblock In {\em Data Privacy Management, Cryptocurrencies and Blockchain
  Technology: ESORICS 2021 International Workshops, DPM 2021 and CBT 2021,
  Darmstadt, Germany, October 8, 2021, Revised Selected Papers}, page
  192–207, Berlin, Heidelberg, 2021. Springer-Verlag.
\newblock \href {https://doi.org/10.1007/978-3-030-93944-1_13}
  {\path{doi:10.1007/978-3-030-93944-1_13}}.

\bibitem{CUP_AdHoc-Now2004}
David Cavin, Yoav Sasson, and Andr{\'{e}} Schiper.
\newblock Consensus with unknown participants or fundamental self-organization.
\newblock In Ioanis Nikolaidis, Michel Barbeau, and Evangelos Kranakis,
  editors, {\em Ad-Hoc, Mobile, and Wireless Networks: Third International
  Conference, {ADHOC-NOW} 2004, Vancouver, Canada, July 22-24, 2004.
  Proceedings}, volume 3158 of {\em Lecture Notes in Computer Science}, pages
  135--148. Springer, 2004.
\newblock \href {https://doi.org/10.1007/978-3-540-28634-9\_11}
  {\path{doi:10.1007/978-3-540-28634-9\_11}}.

\bibitem{concurrentAysnchBC_Cohen_2023}
Ran Cohen, Pouyan Forghani, Juan Garay, Rutvik Patel, and Vassilis Zikas.
\newblock Concurrent asynchronous byzantine agreement in expected-constant
  rounds, revisited.
\newblock Cryptology ePrint Archive, Paper 2023/1003, 2023.
\newblock \url{https://eprint.iacr.org/2023/1003}.
\newblock URL: \url{https://eprint.iacr.org/2023/1003}.

\bibitem{simpleRBC_Crain_2020}
Tyler Crain.
\newblock A simple and efficient asynchronous randomized binary byzantine
  consensus algorithm, 2020.
\newblock \href {http://arxiv.org/abs/2002.04393} {\path{arXiv:2002.04393}}.

\bibitem{DolevDS87}
Danny Dolev, Cynthia Dwork, and Larry~J. Stockmeyer.
\newblock On the minimal synchronism needed for distributed consensus.
\newblock {\em J. {ACM}}, 34(1):77--97, 1987.

\bibitem{DolevLPSW86}
Danny Dolev, Nancy~A. Lynch, Shlomit~S. Pinter, Eugene~W. Stark, and William~E.
  Weihl.
\newblock Reaching approximate agreement in the presence of faults.
\newblock {\em J. {ACM}}, 33(3):499--516, 1986.
\newblock \href {https://doi.org/10.1145/5925.5931}
  {\path{doi:10.1145/5925.5931}}.

\bibitem{Impossibility_Fischer_ACM_1985}
Michael~J. Fischer, Nancy~A. Lynch, and Michael~S. Paterson.
\newblock Impossibility of distributed consensus with one faulty process.
\newblock {\em J. ACM}, 32(2):374–382, apr 1985.
\newblock \href {https://doi.org/10.1145/3149.214121}
  {\path{doi:10.1145/3149.214121}}.

\bibitem{noPrivateSetupBC_Gaoz_IEEE_2022}
Yingzi Gao, Yuan Lu, Zhenliang Lu, Qiang Tang, Jing Xu, and Zhenfeng Zhang.
\newblock Efficient asynchronous byzantine agreement without private setups.
\newblock In {\em 2022 IEEE 42nd International Conference on Distributed
  Computing Systems (ICDCS)}, pages 246--257, 2022.
\newblock \href {https://doi.org/10.1109/ICDCS54860.2022.00032}
  {\path{doi:10.1109/ICDCS54860.2022.00032}}.

\bibitem{AbstractMAC_unreliableLink_PODC2014}
Mohsen Ghaffari, Erez Kantor, Nancy~A. Lynch, and Calvin~C. Newport.
\newblock Multi-message broadcast with abstract {MAC} layers and unreliable
  links.
\newblock In Magn{\'{u}}s~M. Halld{\'{o}}rsson and Shlomi Dolev, editors, {\em
  {ACM} Symposium on Principles of Distributed Computing, {PODC} '14, Paris,
  France, July 15-18, 2014}, pages 56--65. {ACM}, 2014.
\newblock \href {https://doi.org/10.1145/2611462.2611492}
  {\path{doi:10.1145/2611462.2611492}}.

\bibitem{Greve_DSN_2007}
Fabiola Greve and Sebastien Tixeuil.
\newblock Knowledge connectivity vs. synchrony requirements for fault-tolerant
  agreement in unknown networks.
\newblock In {\em 37th Annual IEEE/IFIP International Conference on Dependable
  Systems and Networks (DSN'07)}, pages 82--91, 2007.
\newblock \href {https://doi.org/10.1109/DSN.2007.61}
  {\path{doi:10.1109/DSN.2007.61}}.

\bibitem{AbstractMAC_probabilistic_AdHoc2014}
Majid Khabbazian, Dariusz~R. Kowalski, Fabian Kuhn, and Nancy~A. Lynch.
\newblock Decomposing broadcast algorithms using abstract {MAC} layers.
\newblock {\em Ad Hoc Networks}, 12:219--242, 2014.
\newblock \href {https://doi.org/10.1016/j.adhoc.2011.12.001}
  {\path{doi:10.1016/j.adhoc.2011.12.001}}.

\bibitem{Wattenhofer_ByzCUPs_PODC20}
Pankaj Khanchandani and Roger Wattenhofer.
\newblock Brief announcement: Byzantine agreement with unknown participants and
  failures.
\newblock In {\em Proceedings of the 39th Symposium on Principles of
  Distributed Computing}, PODC '20, page 178–180, New York, NY, USA, 2020.
  Association for Computing Machinery.
\newblock \href {https://doi.org/10.1145/3382734.3405740}
  {\path{doi:10.1145/3382734.3405740}}.

\bibitem{AbstractMAC_kuhn_lynch_newport_2011}
Fabian Kuhn, Nancy Lynch, and Calvin Newport.
\newblock The abstract mac layer, 2011.
\newblock \href {https://doi.org/10.1007/s00446-010-0118-0}
  {\path{doi:10.1007/s00446-010-0118-0}}.

\bibitem{AbstractMAC_DISC2009}
Fabian Kuhn, Nancy~A. Lynch, and Calvin~C. Newport.
\newblock The abstract {MAC} layer.
\newblock In Idit Keidar, editor, {\em Distributed Computing, 23rd
  International Symposium, {DISC} 2009, Elche, Spain, September 23-25, 2009.
  Proceedings}, volume 5805 of {\em Lecture Notes in Computer Science}, pages
  48--62. Springer, 2009.
\newblock \href {https://doi.org/10.1007/978-3-642-04355-0\_9}
  {\path{doi:10.1007/978-3-642-04355-0\_9}}.

\bibitem{Sapta_churn}
Saptaparni Kumar and Jennifer~L. Welch.
\newblock Byzantine-tolerant register in a system with continuous churn.
\newblock {\em CoRR}, abs/1910.06716, 2019.
\newblock URL: \url{http://arxiv.org/abs/1910.06716}, \href
  {http://arxiv.org/abs/1910.06716} {\path{arXiv:1910.06716}}.

\bibitem{LamportSP82}
Leslie Lamport, Robert~E. Shostak, and Marshall~C. Pease.
\newblock The byzantine generals problem.
\newblock {\em {ACM} Trans. Program. Lang. Syst.}, 4(3):382--401, 1982.
\newblock URL: \url{http://doi.acm.org/10.1145/357172.357176}, \href
  {https://doi.org/10.1145/357172.357176} {\path{doi:10.1145/357172.357176}}.

\bibitem{Lynch96}
Nancy~A. Lynch.
\newblock {\em Distributed Algorithms}.
\newblock Morgan Kaufmann, 1996.

\bibitem{AbstractMAC_LEMIS_FOMC2012}
Nancy~A. Lynch, Tsvetomira Radeva, and Srikanth Sastry.
\newblock Asynchronous leader election and {MIS} using abstract {MAC} layer.
\newblock In Fabian Kuhn and Calvin~C. Newport, editors, {\em FOMC'12, The
  Eighth {ACM} International Workshop on Foundations of Mobile Computing (part
  of {PODC} 2012), Funchal, Portugal, July 19, 2012, Proceedings}, page~3.
  {ACM}, 2012.
\newblock \href {https://doi.org/10.1145/2335470.2335473}
  {\path{doi:10.1145/2335470.2335473}}.

\bibitem{Raynal_PODC14_optimalAsyncConsensus}
Achour Mostefaoui, Hamouma Moumen, and Michel Raynal.
\newblock Signature-free asynchronous byzantine consensus with t \&lt; n/3 and
  o(n2) messages.
\newblock In {\em Proceedings of the 2014 ACM Symposium on Principles of
  Distributed Computing}, PODC '14, page 2–9, New York, NY, USA, 2014.
  Association for Computing Machinery.
\newblock \href {https://doi.org/10.1145/2611462.2611468}
  {\path{doi:10.1145/2611462.2611468}}.

\bibitem{AbstractMAC_randomizedConsensus_DISC2018}
Calvin Newport and Peter Robinson.
\newblock Fault-tolerant consensus with an abstract {MAC} layer.
\newblock In Ulrich Schmid and Josef Widder, editors, {\em 32nd International
  Symposium on Distributed Computing, {DISC} 2018, New Orleans, LA, USA,
  October 15-19, 2018}, volume 121 of {\em LIPIcs}, pages 38:1--38:20. Schloss
  Dagstuhl - Leibniz-Zentrum f{\"{u}}r Informatik, 2018.
\newblock \href {https://doi.org/10.4230/LIPIcs.DISC.2018.38}
  {\path{doi:10.4230/LIPIcs.DISC.2018.38}}.

\bibitem{AbstractMAC_consensus_PODC2014}
Calvin~C. Newport.
\newblock Consensus with an abstract {MAC} layer.
\newblock In Magn{\'{u}}s~M. Halld{\'{o}}rsson and Shlomi Dolev, editors, {\em
  {ACM} Symposium on Principles of Distributed Computing, {PODC} '14, Paris,
  France, July 15-18, 2014}, pages 66--75. {ACM}, 2014.
\newblock \href {https://doi.org/10.1145/2611462.2611479}
  {\path{doi:10.1145/2611462.2611479}}.

\bibitem{RandomizedBG_Rabin_SFCS_1983}
Michael~O. Rabin.
\newblock Randomized byzantine generals.
\newblock {\em 24th Annual Symposium on Foundations of Computer Science (sfcs
  1983)}, pages 403--409, 1983.

\bibitem{AuthenticatedFT_Srikanth_1987}
T.~K. Srikanth and Sam Toueg.
\newblock Simulating authenticated broadcasts to derive simple fault-tolerant
  algorithms.
\newblock {\em Distrib. Comput.}, 2(2):80–94, jun 1987.
\newblock \href {https://doi.org/10.1007/BF01667080}
  {\path{doi:10.1007/BF01667080}}.

\bibitem{Toueg_PODC_1984}
Sam Toueg.
\newblock Randomized byzantine agreements.
\newblock In {\em Proceedings of the Third Annual ACM Symposium on Principles
  of Distributed Computing}, PODC '84, page 163–178, New York, NY, USA, 1984.
  Association for Computing Machinery.
\newblock \href {https://doi.org/10.1145/800222.806744}
  {\path{doi:10.1145/800222.806744}}.

\bibitem{Tseng_PODC22_MAC}
Lewis Tseng and Qinzi Zhang.
\newblock Brief announcement: Computability and anonymous storage-efficient
  consensus with an abstract mac layer.
\newblock In {\em {PODC} '22: {ACM} Symposium on Principles of Distributed
  Computing, Italy, 2022}. {ACM}, 2022.

\end{thebibliography}

\appendix

\section{Proof of Termination of MAC-BAC}
\label{app:mac-bac-termination}

$p_{end}$ is defined in equation \ref{eq:p-end}. Termination for MAC-BAC is proven by induction on $p \geq 0$. The base case for $p=0$ holds since every node is initialized with $p_i = 0$.

Inductive Hypothesis: all fault-free nodes proceed to phase $p' \geq 0$ within finite time. 

Now consider phase $p$. When every node is in phase $\geq p$, each node has received $4f + 2$ messages with phase $p$. Thus, due to the inductive hypothesis, each node will update to phase $p + 1$ by line 8. 

\section{Proof of Validity of MAC-BAC}
\label{app:mac-bac-validity}

Fix a round $p$. After each round of the algorithm, the state of each fault-free node in round $p + 1$ must remain in the convex hull of the states of the fault-free nodes from round $p$. That is, the set of states of fault-free nodes in phase $p + 1$ is in a convex hull of the values in phase $p$ for all rounds $p$.




Take a random fault-free node $i$ in phase $p$ which is updating its phase to phase $p+1$. For this to occur, we know that node $i$ has received at least $f+1$ phase $p$ states.


$v_i$ is computed using the average of the $f+1$-th largest  value in $R_i[p_i]$ and the $f+1$-th smallest value in $R_i[p_i]$. There are at least $4f + 2$ values within $R_i[p_i]$, and there are at most $f$ byzantine nodes, so we know that these two values are within the range of all received fault-free states. Thus, $v_i$ must be within the range of all received fault-free states as well for every node $i$. 

So, we know that for a round $p$, the states of all nodes in round $p+1$ must be within the convex hull of values in round $p$.

Therefore, as all rounds of values remain within the convex hull of the values from the previous round, we know that the states of all nodes in round $p_{end}$ will be within the range of the fault-free inputs and will satisfy validity. 



\section{Proof of Lemma \ref{lemma:f1}}
\label{app:range}

\begin{proof}
    \begin{itemize}
        \item $x$ is a possible state when a first mover receives at least $2f + 1$ $x$'s and at most $f$ non-$x$'s. Note that Byzantine node can send any value.

        \item $y$ is a possible state when a first mover receives at least $2f + 1$ $y$'s and at most $f$ non-$y$'s. Note that Byzantine node can send any value.
    \end{itemize}
\end{proof}

\section{Proof of BC-Termination for MAC-RBC}
\label{app:bc-termination}

\begin{lemma}
    Each fault-free node outputs a value with probability 1.
\end{lemma}

\begin{proof}
    First, we prove that no fault-free node remains at a wait statement forever during any phase $p$ -- neither at the wait statement at Line 4, not the wait statement at Line 9.

    To show that no fault-free node is blocked forever at Line 4, we demonstrate that the predicate of Line 4 will eventually become true at every fault-free node $i$. All nodes mac-broadcast their estimate value $v$ at Line 3. By the construction of the background handler for $EST$ messages, all nodes which receive the mac-broadcasts with estimate value $v$ will propagate this message on Line 22. There are at least $4f +1$ fault-free nodes, so $f +1$ of at least one value $v \in \{0, 1\}$ must be received by at least $f + 1$ distinct nodes. Eventually, this value $v$ will be added to $estValues_i$ on Line 25 because all nodes must eventually receive and broadcast $v$ by the guarantees of the abstract MAC layer and Line 22. Therefore, the set $estValues_i$ must eventually be non-empty for all nodes in phase $p$, satisfying the predicate on Line 4.

    To show that no fault-free node is blocked forever at Line 9, we demonstrate that the predicate of Line 9 will eventually become true at every fault-free node $i$, meaning that the Condition WAIT is satisfied. 
    \begin{itemize}
        \item Condition WAIT (1, 2) are eventually satisfied because every node will progress to Line 8, and by the guarantees of the abstract MAC layer, the mac-broadcast on Line 8 will eventually return an ACK at every fault-free node $i$;  thus, $X$ will eventually have $2f + 1$ elements.

        \item Condition WAIT (3) must hold true by the nature of the mac-broadcast on Line 3. All nodes which send $(COMPLETE, p)$ messages must have received their ACK from the mac-broadcast on Line 3. All nodes $x \in X$ must have the same value $v$. This is because there are at least $5f + 1$ nodes within the system. Each node which passes Line 4 will broadcast either one ($v$) or two $(v$ and $v')$ $AUX$ messages on on Line 6. In order for a value to be added to $estValues_i[p]$, an $EST$ message containing this value must have been sent from at least $f + 1$ distinct nodes. 
        There are two cases:
        \begin{itemize}
            \item [(1)] Let time $T_{(3)}$ be the time that $i$ receives its $2f + 1$-st $(COMPLETE, p)$ message corresponding to an $(AUX, v, p)$ message sent by a node $x \in X$. If $estValues_i[p]$ = \{v\}, then there must have been fewer than $f + 1$ nodes who mac-broadcast $v'$ on Line 3 by time $T_{(3)}$. 
            It follows that either (a) there must have been at least $4f + 1$ nodes which mac-broadcast value $v$ at Line 3 by time $T_{(3)}$, hence $X$ has at least $4f + 1$ nodes with value $v$; or (b) there were at least $2f + 1$ nodes which mac-broadcast value $v$ at Line 3 by time $T_{(3)}$, and any other nodes which broadcast value $v'$ will do so only after time $T_{(3)}$.

            \item[(2)] If $estValues_i[p] = \{v, v'\}$, then at time $T_{(3)}$, there must have been at least $f + 1$ nodes which mac-broadcast $v$ at Line $3$, and at least $f + 1$ nodes which mac-broadcast $v'$ at Line $3$. There are $5f + 1 - 2(f + 1) = 3f - 1$ other nodes within the system, so it is possible for $2f + 1$ nodes to broadcast $v$, and another $2f + 1$ nodes to broadcast $v'$. Then, $estValues_i[p] = \{v, v'\}$.
        \end{itemize}

        \item Condition WAIT (4, 5) are eventually satisfied because at most $f$ nodes are Byzantine, so if node $i$ has seen messages from $|U_i|$ other nodes at Line 6, it is guaranteed that node $i$ can receive at least $|U_i| -f$ values which are added to set $Y$.

        \item Condition WAIT (6) is given by construction of $X$ and $Y$.

        \item Condition WAIT (7, 8) must be true by construction of $value_i[p]$ since all values received from $Y$'s $(AUX, *, p)$ messages must have been elements of $estValues[p]$.
    \end{itemize}
    Therefore, all predicates of the Condition WAIT must be satisfied for all nodes in phase $p$, satisfying the predicate on Line 9.

    \begin{lemma}
    \label{claim:probability1}
        With probability $1$, there is a phase $p$ in which all fault-free nodes have the same estimate value.
    \end{lemma}

    \begin{proof}
        Fix a phase $r$. There are three possible scenarios:
        \begin{itemize}
            \item[(1)] If the predicate on Line 11 evaluates to true at all fault-free nodes in phase $r$, then there are two cases. (i) If the predicate of Line 12 evaluates to true, all fault-free nodes execute Line 13, they will have output $v$, and their estimate value remains equal to $v$ which is also the value of the common coin in phase $r$; (ii) Otherwise, all fault-free nodes will continue to execute Line 15 and will set their estimate value to equal the value of the common coin in phase $r$. Note that by the assumption of common coin, all fault-free nodes will see the same value $c$; hence, they will take the same action in this case.

            \item[(2)] If all fault-free nodes execute Line 17, then they will all take the value of the common coin in phase $r$ as their estimate value for phase $r + 1$, and the claim directly follows.
            
            \item[(3)] Let behavior $output$ be the scenario at node $i$ if the predicate of Lines 11 and 12 both evaluate to true and node $i$ outputs value $v = c$. Let behavior $continue$ be the scenario at node $i$ if the predicate of Line 11 evaluates to true, but the predicate of Line 12 evaluates to false, that is, $v \neq c$. 
            
            By Lemma \ref{lemma:v=u}, $values_i[r] = values_j[r] = \{v\}$. for all nodes at which the predicate of Line 11 evaluates to true. Therefore, either all nodes at which the predicate of Line 11 evaluates to true will $output$, or all will $continue$. 
            Meanwhile, all other nodes for which the predicate of Line 11 evaluates to false will execute Line 17 and take the value of the common coin. The value of the common coin at any phase $r$ is independent from its value in any other phase by the properties of the common coin. 
            
            Therefore, $v = c$ (Line 12 will evaluate to true, and nodes will $output$ $v$) with probability $p = \frac{1}{2}$. Similarly, $v \neq c$ and nodes which pass Line 11 will $continue$ with probability $p = \frac{1}{2}$. Let $P(r)$ be the probability that there is a phase $r'$ where $r' < r$ such that the value $v$ in phase $r'$ is equal to the value of $c$ in round $r'$. Then, $P(r) = p + (1 - p)p + \cdots + (1 - p)^{r-1}p = 1 - (1 - p )^{r}$. So, $\lim_{r\rightarrow \infty}P(r) = 1$, proving Claim \ref{claim:probability1}.

        \end{itemize}
    \end{proof}

    \begin{lemma}
\label{lemma:v}
    Fix a phase $p$ at the beginning of which all fault-free nodes have the same estimate value $v$.  
    These nodes will remain with value $v$ forever. That is, $v_i = v$ for all fault-free node $i$ after phase $p$.
\end{lemma}

\begin{proof}
    If all fault-free nodes have the same estimate value $v$ at the beginning of phase $p$, they all mac-broadcast$(EST, v_i, p_i)$ at Line 2. Then, each fault-free node will have $v \in estValues[p_i]$, where $v$ is the only value in $estValues[p_i]$. There are at least $f + 1$ fault-free nodes, so the predicate of Line 20 will eventually evaluate to true. Every fault-free node which passes Line 20 will mac-broadcast an $(EST, v, p)$ message. We know by the guarantees of the mac-broadcast, that all the other fault-free nodes will eventually receive and broadcast any value which is received by $f+1$ fault-free nodes. There are at least $4f + 1$ fault-free nodes within the system, so eventually they will broadcast $(EST, v, p)$ message, satisfying the predicate to Line 23, and $v$ will be added to $estValues_i[p]$.  Therefore, the estimate value $v_i$ is set to $v$ and must remain $v$ in phase $p + 1$ and all subsequent phases.
\end{proof}

    By Lemma \ref{claim:probability1} and Lemma \ref{lemma:v}, it follows that the estimate value $v_i$ of all fault-free nodes $i$ will remain $v_i = v$. Let phase $p$ be the phase in which all fault-free nodes have the same estimate value $v$. Hence, the predicate of Line 11 will evaluate to true at every phase after phase $p$. Following from the common coin, we know that with probability $1$, there is eventually a phase in which $coinflip()$ outputs $v$, so $c = v$. Hence, the predicates of Lines 11 and 12 will evaluate to true, and all fault-free nodes will output $v$ with probability 1.

\end{proof}


\section{Impossibility}
\label{app:impossible}

\subsection{Proof that without the knowledge of $f,$ it is impossible to solve consensus}

\begin{theorem}
    Without the knowledge of $f,$ it is impossible to solve Byzantine approximate consensus.   
\end{theorem}

\begin{proof}
    Assume that there exists a uniform algorithm $A$ that solves consensus for a certain $n$ and $f$, but the algorithm itself does not know these values for $n$ and $f$. An algorithm is uniform if all nodes are using the same algorithm. 
    
    To satisfy $\epsilon$-agreement, all $n-f$ fault-free nodes must have an output such that no two fault-free noes $i$ and $j$ have states $|v_i - v_j| > \epsilon$.
    
    Consider Scenario 1 and Scenario 2 below. In both scenarios, the nodes do not know the exact value of $n$, nor $f$:
    
    1. Scenario 1: Suppose $n = 2$, $f = 0$. Assume two nodes, $a$ and $b$, are fault-free nodes. A node can only communicate with one other node to make a decision, since $n=2$. 
    
    Suppose node $a$ has input $0$ and node $b$ has input $1$.
    
    To satisfy validity, there must be an output which can be any value in the range $[0, 1]$ at $a$ and $b$ such that $\epsilon$-agreement is satisfied. In fact, $a$ can have one value and $b$ can have one value as long as $|v_a - v_b| \leq \epsilon$.
    
    2. Scenario 2: Assume there are any number of nodes within the system satisfying $n \geq 2f + 1$. In this scenario, a node, say node $a$, can only communicate with one other node to make a decision. This is because to node $a$, it is indistinguishable whether $a$ is in Scenario 1 or Scenario 2 since there is no way for node $a$ to know the values of $n$ or $f$. 
    
    We can impose a delay $D$ on all messages other than those from node $b$. 
    
    Similar to the setup in \cite{Sapta_churn}, assume that $D$ is an unknown upper bound, on the delay of any message sent between between fault-free nodes. Let $D = t' - t$ where $D > 0$ be the maximum delay on a message sent at time $t$ and received at time $t'$.
    
    For two fault-free nodes $a$ and $b$, if a message sent from node $a$ at time $t$ and node $b$ is active throughout $[t, t + D]$, then we know that node $b$ will receive the message from $a$, and that the delay of every received message is in the range $(0, D]$.
    
    Assume $(0, D]$ is the bound on the delay of the message sent from node $b$. In the time interval $(0, D]$, assume the only message which node $a$ receives is a message from node $b$. During this time interval, node $a$ cannot distinguish between Scenarios 1 and 2.
    
    Run $A$ in Scenario 2, within the time interval $(0, D]$, assuming the Byzantine node $b$ behaves in the exact same manner as it did in Scenario 1 to node $a$. Once again, from the perspective of node $a$, Scenarios 1 and 2 remain indistinguishable and there is no way for node $a$ to know whether the behavior from $b$ is Byzantine or not. Node $a$ can output a value in $(0, 1]$, but this is a contradiction because any value in $(0, 1]$ is not within the range of fault-free inputs, $[0]$, so validity is violated.    
    
    Next, run $A$ in Scenario 2, this time assuming the Byzantine node $a$ behaves in the exact same manner as it did in Scenario 1 to node $b$. Once again, from the perspective of node $b$, Scenarios 1 and 2 remain indistinguishable and there is no way for node $b$ to know whether the behavior from $a$ is Byzantine or not. Node $b$ can output a value in $[0, 1)$, but this is a contradiction because any value in $[0, 1)$ is not within the range of fault-free inputs, $[1]$, so validity is violated. 
\end{proof}

\end{document}